\documentclass[runningheads]{llncs}
\usepackage{geometry}
\usepackage{graphicx}
\usepackage{tikz}
\usepackage{academicons}

\usepackage[hidelinks]{hyperref}

\usepackage{mathtools}

\usepackage{amssymb}
\usepackage{url}
\usepackage{enumitem}
\usepackage{graphicx}
\usepackage{xspace}
\usepackage{caption}
\usepackage{multirow}
\usepackage[normalem]{ulem}

\usepackage{caption}
\usepackage{subcaption}

\newcommand{\opt}{\textsc{Opt}\xspace}
\newcommand{\Alg}{\textsc{\mbox{Harmonic-Stretch}}\xspace}
\newcommand{\HS}{\textsc{Harmonic-Shifting}\xspace}

\newcommand{\FF}{\textsc{First-Fit}\xspace}

\newcommand{\CR}{\textsc{$1.75$}\xspace}

\usepackage{subcaption}
\begin{document}
\input{figs}
\title{On the Fault-Tolerant \\ Online Bin Packing Problem}
%
%
\author{Shahin Kamali
\and Pooya Nikbakht
}
%
\institute{
\ \\ Department of Computer Science, University of Manitoba, Winnipeg, Canada\\
\email{shahin.kamali@umanitoba.ca},
\email{nikbakhp@myumanitoba.ca} \ \\ \
}

\maketitle              
\begin{abstract}
	We study the fault-tolerant variant of the online bin packing problem. Similar to the classic bin packing problem, an online sequence of items of various sizes should be packed into a minimum number of bins of uniform capacity. For applications such as server consolidation, where bins represent servers and items represent jobs of various loads, it is necessary to maintain fault-tolerant solutions. In a fault-tolerant packing, any job is replicated into $f+1$ servers, for some integer $f>1$, so that the failure of up to $f$ servers does not interrupt service.
    We build over a practical model introduced by Li and Tang [SPAA 2017] in which each job of load $x$ has a primary replica of load $x$ and $f$ standby replicas, each of load $x/\eta$, where $\eta >1$ is a parameter of the problem. Upon failure of up to $f$ servers, any primary replica in a failed bin should be replaced by one of its standby replicas so that the extra load of the new primary replica does not cause an overflow in its bin. \\
    We study a general setting in which bins might fail while the input is still being revealed.
    Our main contribution is an algorithm, named \Alg, which maintains fault-tolerant packings under this general setting.
    We prove that \Alg has an asymptotic competitive ratio of at most 1.75. This is an improvement over the best existing asymptotic competitive ratio 2 of an algorithm by Li and Tang [TPDS 2020], which works under a model that assumes bins fail only after all items are packed.
	 	
\end{abstract}

\vspace*{.3mm}
\section{Introduction}
In the classic online bin packing problem, a sequence of items of various sizes in the range $(0,1]$ is revealed in an online manner. The goal is to pack these items into a minimum number of bins of uniform capacity 1. An online algorithm places each item into a bin without any information about the forthcoming items.
For example, the \FF algorithm places each item into the first bin that has enough space (and opens a new bin if no bin has enough space). The Harmonic family of algorithms place items of similar sizes together in the same bins (see \cite{LeeLee85,RamBrowLeeLee89,Seid02} for some variants of the Harmonic algorithm). 

Server consolidation is an application of the online bin packing problem in the cloud. Service providers such as Amazon EC2~\cite{EC2} or Microsoft Azure~\cite{Azure} host client applications, called \emph{tenants}, on their servers. 
Upon the arrival of a tenant, an online algorithm assigns it to a server that has enough resources available for the tenant. 
For each tenant, the service provider commits to a Service Level Agreement (SLA) that specifies the minimum performance requirement for the tenant~\cite{articleSLA}. In particular, different tenants have various loads that are indicated in their SLAs. 
The goal of the service provider is to satisfy the SLA requirements while minimizing operational costs. To achieve this goal, cloud providers often \emph{consolidate} tenants on shared computing machines to improve utilization~\cite{AjiroT07}. 
Such consolidation can be modeled via online bin packing, where each item represents a tenant and each bin represents a computing machine (a server). 

In cloud systems, client applications are often replicated in multiple servers. Such replication is necessary to avoid service interruptions in the case of server failures. In order to keep services uninterrupted against the failure of up to $f$ services, it is necessary to replicate each tenant in at least $f+1$ different servers. In many practical scenarios, each tenant has a \emph{primary} replica, which handles read/write queries, and multiple \emph{standby} replicas, which act as backup replicas in anticipation of server failures. Naturally, the computational resources required for hosting primary replicas are more than that of standby replicas \cite{yanagisawa2013}. 

In this paper, we continue a line of research on the primary-standby scheme for the fault-tolerant bin packing problem~\cite{liTang17,liTang20}. The goal is to pack an online sequence of items (tenants) into a minimum number of bins (servers) such that each item of size $x$ has a primary replica of size $x$ and $f$ standby replicas, each of size $x/\eta$, for some parameter $\eta >1$. Over time, some servers might fail, and some of the previously failed servers might recover. An algorithm has no knowledge about how servers fail or recover, but it is guaranteed that the number of failed servers at each given time is at most $f$. 
To ensure the service is fault-tolerant, the primary replica of each job should be available at any given time. Therefore, when a server that hosts the primary replica of an item $x$ fails, a standby replica of $x$ should be selected to become its new primary replica. The subsequent increase in the load of such replica {(from $x/\eta$ to $x$)} should not cause an overflow in the bin (see Section~\ref{sect:model} for a formal definition).

The \emph{asymptotic competitive ratio} is the standard measure for comparing online bin packing algorithms.
An online algorithm A is said to have an asymptotic competitive ratio $r$ iff for any sequence $\sigma$ we have $A(\sigma) \leq  r \ \opt(\sigma) + c$, where $A(\sigma)$ is the number of bins in the packing of A for $\sigma$, $\opt(\sigma)$ is the number of bins in an optimal packing of $\sigma$, and $c$ is a constant independent of the length of $\sigma$. Throughout the paper, we use the term competitive ratio to refer to the asymptotic competitive ratio.

\subsection{Previous work \& contribution}\label{sect:prev}

Schaffner et al.~\cite{SchaJan13} proposed the first model for the fault-tolerant bin packing problem and studied the competitiveness of a few basic algorithms. Subsequent work on this model resulted in improved algorithms~\cite{DaudjeeKL14,MateDK17,Exper2021}. These initial results studied a model in which the load of an item is evenly distributed between all its replicas. Li and Tang~\cite{liTang17} introduced an alternative model that distinguishes between primary and standby replicas. 
This model was further studied in~\cite{liTang20}, where the \HS algorithm was presented and proved to have a competitive ratio of at most 2. 
All previous algorithms, in particular the \HS algorithm, assume that an online sequence is first packed and \emph{then} a set of up to $f$ bins might fail. In practice, however, the packing is an ongoing process, and the servers might fail while the input is still being revealed. 

In this paper, we study the fault-tolerant bin packing problem with the primary-standby scheme, as introduced in~\cite{liTang17}. We assume that bins can fail and recover in an online manner so that at most $f$ bins are failed at the same time. As such, an algorithm in this model requires a \emph{packing strategy}, which allocates items to bins, and an \emph{adjustment strategy}, which makes necessary adjustments (i.e., promoting a standby replica to a primary replica and vice versa) when bins fail or recover. 
The packing and adjustment strategies should coordinate to maintain ``valid" solutions that are tolerant against bin failures, that is, primary replicas are available for all items, and no bin is overloaded throughout the packing and adjustment processes.

We introduce an algorithm named \Alg that maintains fault-tolerant packings for an online sequence of items. As the prefix ``Harmonic" suggests, \Alg classifies items by their sizes. The classification and treatment of items in each class is, however, different from the existing Harmonic-based algorithms. In particular, unlike the algorithms in~\cite{liTang17,liTang20}, which classify {items} based on the size of their standby replicas, \Alg classifies items based on the size of their both primary and standby {replicas}. The placement and adjustment strategies in \Alg are designed in a flexible way that allows maintaining valid packings even if bins fail before the packing completes. We prove that \Alg has a competitive ratio of at most \CR, which is an improvement over the competitive ratio 2 of \HS of~\cite{liTang20}. In summary, \Alg is designed to work in a more general setting, and yet achieves a better competitive ratio when compared to the previous algorithms.


\section{Primary-standby model for fault-tolerant bin-packing}\label{sect:model} 

The primary-standby scheme for the fault-tolerant bin packing problem is defined as follows:
\vspace*{1mm} 
\begin{definition}\label{def:problem}
In the \emph{$(f,\eta)$-fault-tolerant bin packing problem}, a sequence of $n$ items, each having a size in the range $(0,1]$, is revealed in an online manner. When an item of size $x$ arrives, an algorithm places a \emph{primary replica} of size $x$ and $f$ \emph{standby replicas}, each of size $x/\eta$, into bins of unit capacity, without any prior information about the forthcoming items. 
Throughout the packing process, some bins might fail and some of the previously failed bins might recover, in a way that the number of failed bins stays at most $f$ at any time. 
In a valid packing, a primary replica of each item should be always available. Therefore, upon failure of a bin with a primary replica of an item $x$, a standby replica of $x$ (in a non-failed bin) needs to be selected and promoted to become the new primary replica of $x$. The subsequent increase in the size of the {promoted standby} replica {(from $x/\eta$ to $x$)} should not cause an overload in any bin. The objective is to maintain a valid packing with a minimum number of bins.
\end{definition}

We assume that the packing of the $f+1$ replicas of any item takes place concurrently, that is, no bin fails when a group of $f+1$ replicas is being packed. 
Note that an algorithm can change the status of a replica from primary to standby and vice versa, but it cannot move replicas from one bin to another. 
In order to achieve a valid packing, the $f+1$ replicas of each item need to be packed in $f+1$ different bins; otherwise, failure of up to $f$ bins that contain all replicas of an item makes that item {inaccessible}. 

\begin{example}
Figure~\ref{fig:validpacking} illustrates Definition~\ref{def:problem}. 
Each item of size $x$ has a primary replica of size $x$ and $f=2$ standby replicas of size $x/\eta$, where $\eta = 2.0$. 
The packing~(a) is a valid packing. The failure of any single bin or a pair of bins {at any time} can be addressed by promoting a standby replica into a primary replica {without overloading any bin}. {For example}, the arrows in the figure point to the standby replicas that are selected to become primary replicas after the simultaneous failure of bins $B_1$ and $B_2$ at time $t_1$. The packing~(b), on the other hand, is not a valid packing of $\sigma$: if $B_1$ and $B_2$ fail, it is not possible to select standby replicas to replace the failed primary replicas without overloading a bin.  
\end{example}

\begin{figure}[!h]
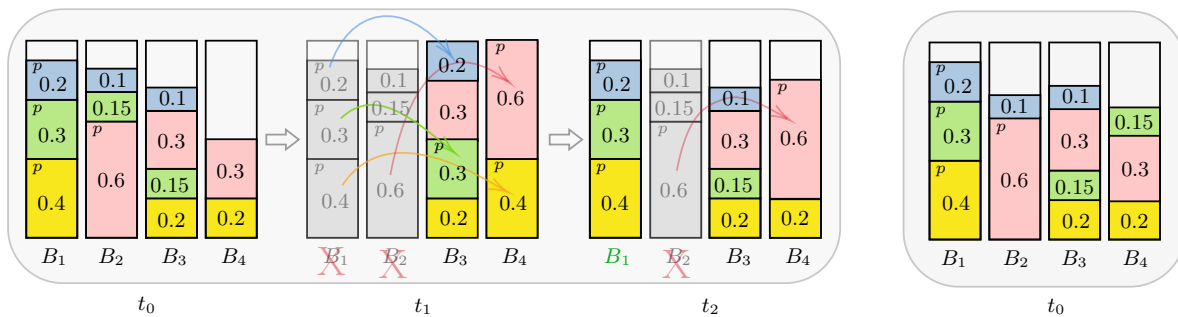

	\centering
	\begin{subfigure}[b]{0.73\textwidth}
		\centering
		\scalebox{.9}{\defPartA}
		\caption{A valid solution: four items are packed at time $t_0$. Bins $B_1$ and $B_2$ fail at time $t_1$, and standby replicas in other bins replace primary replicas in the failed bins. $B_1$ recovers at $t_2$, and the algorithm gives back the primary status to replicas in $B_1$.}
		\label{fig:validSubFig}
	\end{subfigure}
\hspace{0.01\textwidth}
	\begin{subfigure}[b]{0.24\textwidth}
		\centering
		\scalebox{.9}{\defPartB}
		\caption{An invalid solution at time $t_0$.\ \\ \ \\ \  }
		\label{fig:InvalidSubFig}
	\end{subfigure}
	\caption{Packing sequence $\sigma =  (0.4,0.6,0.3,0.2)$, where $f=2$ and $\eta=2.0$. Replicas of the same items have the same color; primary replicas have the letter $p$.}
	\label{fig:validpacking}
\end{figure}

\section{\Alg algorithm}

In this section, we present the \Alg algorithm. First, we provide an overview of the main components of the algorithm.
	\paragraph{Item classification:} Items are partitioned into \emph{classes}, based on the size of their primary and standby replicas. There are 7 possible classes for primary replicas and $\lfloor 6\eta \rfloor+1$ classes for standby replicas. An item that has a primary replica of class $i$ ($1\leq i \leq 7) $ and a standby replica of class $j$ ($1\leq j \leq \lfloor 6\eta \rfloor+1$) is called an $(i,j)$-items. When $i=7$ and $j=\lfloor 6\eta \rfloor+1$, the $(i,j)$-items are called \emph{small items}, while other items are called \emph{regular items} (see Section~\ref{sect:class} for details).\\	
	Items that have the same primary and standby classes are packed separately from other items, that is, a replica of an $(i,j)$-item is never placed with a replica of an $(i',j')$-item together in the same bin if $i\neq i'$ or $j\neq j'$.
	\paragraph{Maintaining bin groups:} To place $(i,j)$-items, the algorithm maintains \emph{groups} of bins. Each group is formed by a constant number of (initially empty) bins. 
	At any given time, there is one ``active" group for $(i,j)$-items, where the incoming $(i,j)$-items are packed into.
	When no more items can fit {into the active group}, the group becomes ``complete", and another group becomes the active group. 
	If a bin of the active group fails, the algorithm {declares that group as ``unavailable",} leaves {that group} ``incomplete" and selects a new active group. When all failed bins of an incomplete group recover, that group becomes ``available" again. When a new active group is required (i.e., when the currently active group becomes complete or one of its bins fails), an {(incomplete)} available group is selected as the new active group. If no available group exists, the algorithm opens a fresh group with empty bins as the active group. 
	Given that at most $f$ bins can fail at the same time, the algorithm maintains up to $f+1$ incomplete groups for $(i,j)$-items, out of which one group is the active group, and $f$ groups are either unavailable or 
	include bins that have recovered from a failure (see Section~\ref{sect:groups}). 
		
	\paragraph{Packing strategy:} When placing $(i,j)$-items inside their active group, primary and standby replicas are packed in separate bins, which are respectively called {\emph{primary} and \emph{standby bins}}. Items in primary bins are packed as tightly as possible, while items in the standby bins are packed so that there is enough space for the promotion of exactly one replica. The packing strategy ensures that a primary bin shares replicas of at most one item with any standby bin. For small items{, a consecutive number} of them are merged to form ``super-replicas"; each super-replica is treated in the same way that regular items are packed (see Section~\ref{sect:packstg}).

	\paragraph{Adjustment strategy:} When a bin $B$ fails, for each primary replica $x$ in $B$, a standby replica of $x$ in a non-failed bin should be promoted to become the new primary replica. 
	This is done through maintaining an injective mapping $h$. The domain of $h$ is the set of primary replicas like $x$ {residing in the} primary bins that are failed, and the range of $h$ is a set of non-failed standby bins in the same group such that $h(x)$ contains a standby replica of $x$.   
	 The injective nature of the mapping, and the fact that there is enough space for expansion of one standby replica in $h(x)$, implies that one can promote {the standby} replica in $h(x)$ to replace $x$ as the primary replica, without causing an overflow in $h(x)$.
	The assignment of standby replicas as primary replicas is temporary, that is, upon the recovery of $B$, its primary replicas like $x$ will retain their primary status and are removed from the domain of $h$, while the promoted replica in $h(x)$ is demoted to become a standby replicas again (see Section~\ref{sect:adj}). 
	
In what follows, we explain the above components in more detail.

\subsection{Item classification}\label{sect:class}
\noindent There are seven classes for primary replicas and $\lfloor 6\eta \rfloor+1$ classes for standby replicas  (see Table~\ref{tbl:weighting} for details). An item has \emph{primary class} $i \in \{1,2,3,4,5\}$ if its primary replica is {of size} in the range $(\frac{1}{i+1}, \frac{1}{i}]$, primary class 6 if its primary replica is {of size} in the range $(\frac{1}{7-1/\eta},\frac{1}{6}]$, and primary class $7$ if its primary replica is of size at most $\frac{1}{7-1/\eta}$. We refer to items of primary class $i \leq 6$ as \emph{regular} items, and items of primary class 7 as \emph{small} items.
An item has \emph{standby class} $j \in \{ 1,2, \ldots, \lfloor 6\eta \rfloor -1 \}$ if its standby replica has size in the range $(\frac{1}{j+\eta}, \frac{1}{j+\eta-1}]$, standby class $j= \lfloor 6\eta \rfloor$ if its standby class has size in the range $(\frac{1}{7\eta-1}, \frac{1}{\lfloor 6\eta \rfloor +\eta -1}]$, and standby class $\lfloor 6\eta \rfloor+1$ if its standby replica is of size $\frac{1}{7\eta-1}$.

In what follows, we refer to an item of primary class $i$ and standby class $j$ as an $(i,j)$-items.
Primary replicas of small items (with $i=7$) are in the range $(0,\frac{1}{7-1/\eta}]$, and their standby replicas are in the range $(0,\frac{1}{7\eta-1}]$. 
Therefore, an $(i,j)$-item is a small item if $i=7$ and $j=\lfloor 6\eta \rfloor+1$, and a regular item otherwise. 

\begin{table*}[!t]
	\begin{center}
		
		\scalebox{.9}{
			\def\arraystretch{1.7}
			\begin{tabular}{|c|c|c|c|}
				\hline
				\multicolumn{4}{|c|}{\textbf{primary replicas}} \\
				\hline
				\textbf{class} & \textbf{size} & \textbf{weight} & \textbf{density} \\ \hline \hline
				
				$i = 1$ & $(\frac{1}{2}, 1]$ & $1$ & $<2$ \\ \hline
				
				$i = 2$ & $(\frac{1}{3}, \frac{1}{2}]$ & $\frac{1}{2}$ & $<\frac{3}{2}$\\ \hline
				
				$\vdots$ & $\vdots$ & $\vdots$ & $\vdots$ \\ \hline
				
				\ $i \in [1, 5]\ $ & $(\frac{1}{i+1}, \frac{1}{i}]$ & $\frac{1}{i}$ & $<\frac{i+1}{i}$\\ \hline
				
				$i = 6$ & \ $(\frac{1}{7-1/\eta}, \frac{1}{6}]$ \ & $\frac{1}{6}$ & $<\frac{7}{6}$\\ \hline
				
				$i=7$ & $p \in (0, \frac{1}{7-1/\eta}]$ & $\frac{3}{2}p$ & $\frac{3}{2}$\\ \hline

			\end{tabular}
			\ \ \ \ \ \ 
			\begin{tabular}{|c|c|c|c|}
				\hline
				\multicolumn{4}{|c|}{\textbf{standby replicas}} \\
				\hline
				\textbf{class} & \textbf{size} & \textbf{weight} & \textbf{density} \\ \hline \hline
				
				$j=1$ &  $(\frac{1}{\eta+1}, \frac{1}{\eta}]$ & 1 & $<\eta+1$\\ \hline
				
				$j=2$ &  $(\frac{1}{\eta+2}, \frac{1}{\eta+1}]$ & $\frac{1}{2}$ & $<\frac{\eta+2}{2}$\\ \hline
				
				$\vdots$ &  $\vdots$ & $\vdots$ & $\vdots$\\ \hline
				
				\ $j \in [1, \lfloor 6\eta \rfloor - 1]\ $ &  $(\frac{1}{\eta+j}, \frac{1}{\eta+j-1}]$ & $\frac{1}{j}$ & $<\frac{\eta+j}{j}$\\ \hline
				
				$j = \lfloor 6\eta \rfloor$ &  \ $(\frac{1}{7\eta-1}, \frac{1}{\eta+\lfloor 6\eta \rfloor-1}]$ \ & $\frac{1}{\lfloor 6\eta \rfloor}$ & $<\frac{\eta+\lfloor 6\eta \rfloor}{\lfloor 6\eta \rfloor}$\\ \hline
				
				$j=\lfloor 6\eta \rfloor + 1$ &  $s\in (0, \frac{1}{7\eta-1}]$ & $\frac{3}{2}s$ & $\frac{3}{2}$\\ \hline
			\end{tabular}
		}
	\end{center}
	\caption{A summary of the {replica classes} used in the definition and analysis of the \Alg algorithm. The weight and density of classes is used in the analysis of the algorithm.}
	\label{tbl:weighting}
\end{table*}

\subsection{Maintaining bin groups}\label{sect:groups} For each pair of $i,j$, the algorithm maintains one group of bins, which are all non-failed, {as the \emph{active} group}. As $(i,j)$-items are revealed, they are packed into bins of the active group, as will be explained in Section~\ref{sect:packstg}. 
In the beginning, a group of all-empty bins is opened and declared as the active group for $(i,j)$-items. 
A new active group is needed when either i) one of the bins in the active group fails or ii) enough items are placed inside the active group, and the group becomes \emph{complete};  before that, the group is \emph{incomplete}.  
A group is said to be \emph{available} if none of its bins are failed. An active group is always available. 
When a new active group is required, the algorithm checks whether an incomplete and available group exists. Such a group, if it exists, is a former active group that, at some point lost its active status due to a bin failure. Since the group is now available, its failed bins are now recovered. {If such a group exists, it is selected as the new active group (if multiple such groups exist, one is chosen arbitrarily).} On the other hand, if no incomplete, available group exists, the algorithm opens a fresh group of all-empty bins and declares it as the active group.
 
\begin{lemma}
There are at most $f+1$ incomplete groups at each given time during the execution of the algorithm.\label{lem:new}
\end{lemma}
\begin{proof}
Consider otherwise, that is, at some point, there are at least $f+2$ incomplete groups. Let $t$ denote the time at which the $(f+2)$'th group $G$ is initiated. There are at most $f$ groups that contain at least one failed bin at any given time, in particular, at time $t$. So, out of the $f+1$ incomplete groups at time $t$ (before $G$ is initiated), at least one group $G'$ has been incomplete and available. Therefore, $G'$ had to be selected as the new active group instead of $G$, a contradiction.\qed
\end{proof}

\subsection{Packing strategy}\label{sect:packstg}

\noindent 

We explain how the algorithm packs $(i,j)$-items inside their active group. The placement is slightly different for regular and small items:
\vspace*{2mm}

\noindent \textbf{Regular items.}
We describe how $(i,j)$-items are packed, where $i \leq 6$ and $j \leq \lfloor 6\eta\rfloor$. 
Each bin group for $(i,j)$-items, in particular the active group, is formed by $j + fi$ bins and has enough space for $ij$ items. The group becomes complete when $ij$ items are placed in it. 
There are $j$ primary bins $B_0, B_1, \ldots, B_{j-1}$ that are each partitioned into $i$ \emph{spots} of capacity $1/i$.
There are $fi$ standby bins formed by $f$ sets of bins, each containing $i$ standby bins. We use $\beta^k_0, \beta^k_1, \ldots \beta^k_{i-1}$ to denote the standby bins in the $k$'th set ($k\leq f$). Each standby bin is partitioned into $j$ spots of size $\frac{1}{j+\eta-1}$. This leaves a \emph{reserved space} of size $\frac{\eta-1}{j+\eta-1}$ in the bin. The spots in the standby bins are labeled from $0$ to $j-1$.

Let $a_t$ be the $t$'th item that is to be packed into the group ($0 \leq t \leq ij-1$). Let $w~=~(t~\mod~j)$ and $z = \lfloor t/j \rfloor$. 
Note that $w$ and $z$ are in the ranges $[0,j-1]$ and $[0,i-1]$, respectively. 
The algorithm places the primary replica of $a_t$ in the spot $z$ of the primary bin $B_w$ {in the active group}. Standby replicas of $a_t$ are placed in the spot $w$ of bins $\beta^1_z, \beta^2_z, \ldots , \beta^f_z$ {of the active group}.

\begin{example}
Figure~\ref{fig:algExp} illustrates packing items of class $(i,j)$, where $i=2$, $j=5$, and $f=3$ in a complete group.
There are $j=5$ primary bins $B_0, \ldots, B_4$, each partitioned into $i=2$ spots. 
There are {$f=3$} groups of standby bins, each containing $i=2$ bins that are partitioned into $j=5$ spots. 
For an item like $a_4$ (the red item), we have $w=4$ and $z=0$. The primary replica of $a_4$ is thus placed in the $0$'th spot of the $4$'th bin, while standby replicas of $a_4$ are placed in the $4$'th spot of the bin $\beta^k_0$  for {$k \in [1,3]$}. 
\end{example}

\begin{figure*}[!t]
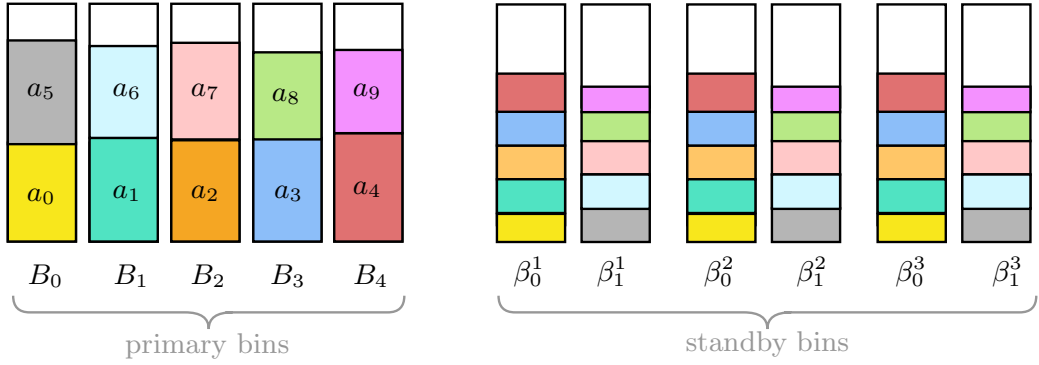

	\begin{center}
		\scalebox{1.2}{\algfigoct}
		\caption{An illustration of the \Alg packing for regular {$(i,j)$-}items {in a complete group}, where $i=2$, $j=5$, and {$f=3$}.} 
		\label{fig:algExp}
	\end{center}
\end{figure*}

\begin{example}
Figure~\ref{fig:apx} shows an incomplete group of bins opened for $(i,j)$-items where $i=2$, $j=5$, and  $f=3$. The group includes $j+fi=11$ bins. Each of the primary bins and standby bins has $i=2$ and $j=6$ spots, respectively, out of which the black spots have not been filled yet. The algorithm has placed items $a_0, \ldots a_6$ in their respective bins when the group was active (and {, hence,} available). At some time $t$, bins $B_0$ and $B_4$ failed. At this point, the group becomes unavailable, and the adjustment strategy processes $a_0$, $a_5$, and $a_4$ to assign new primary replicas for them. After time $t$, the group will not be active anymore { (because it is not available)}, and the algorithm does not place upcoming items in this group until it is selected as the active group later again; this requires the group to become available again, that is, $B_0$ and $B_4$ recover. 
\end{example}
\ \\ \ \\
\noindent \textbf{Small items.}
In order to pack replicas of small items, \Alg merges sets of {consecutive} small items into \emph{super-replicas} (SRs). 
Given that the size of small primary and standby replicas are respectively at most $\frac{1}{7-1/\eta}$ and $\frac{1}{7\eta -1}$, it is possible to group consecutive primary and standby replicas into SRs with sizes in the range $(\frac{1}{7-1/\eta},\frac{2}{7-1/\eta})$ and $(\frac{1}{7\eta -1},\frac{2}{7\eta -1})$, respectively. 
The algorithm maintains an (initially empty) \emph{open} primary SR of {capacity} $\frac{2}{7-1/\eta}$ and places consecutive primary replicas in the open SR until placing the next replica causes the total size of replicas in the SR {to exceed} its capacity. At this point, the SR is closed and a new SR is opened. 
Similarly, the algorithm maintains $f$ (initially empty) open standby SRs, each of capacity $\frac{2}{7\eta-1}$, and places the $f$ standby replicas in these bins until a replica does not fit the open SRs, at which point the~$f$ open SRs are closed and {a set of $f$} new SRs are opened. Since the primary and standby SRs are opened and closed at the same time, we can think of a set of small replicas that are placed in an SR as a single {regular} replica.
In what follows, we describe how the newly opened SRs are placed into bins.

Each group $G$ of bins opened for small items contains $f$ standby bins that mirror each other\footnote{Two bins ``mirror" each other iff they contain the same set of replicas.} and one primary replica which is ``committed" to $G$. There are also ``free" primary bins that are not committed to any group.
Upon the arrival of a small item, its standby replicas are placed into the open SRs located on (mirroring bins) of the active group, and its primary replica is placed into the open SR located on the committed bin. As before, if any bin of the active group fails, the group becomes unavailable, and the algorithm declares another incomplete, available group as the active group (or creates a new one if no such group exists). 
 There is a reserved space of size $\frac{2(\eta-1)}{7\eta -1}+\frac{2}{7\eta -1}=\frac{2(\eta-1)}{7\eta -1}$ inside each of the $f$ mirroring bins of a group, which is used for the promotion of the standby SRs when required. 
When it is needed to open a new SR (when the new replicas do not fit in the open SR), \Alg first places $f$ standby SRs and then a single primary SR as follows.
For the $f$ standby SRs, if the available space in the mirroring standby bins of the active group is at least $ \frac{2\eta}{7\eta -1}$, 
then the new SRs will be placed into the existing open standby bins. Otherwise, the active group gets complete, and {either} another incomplete, {available} group is selected as the active group, {or (if there is no such an available group) a new group is opened.} For placing a new primary SR, the algorithm first frees the bin committed to the active group, and then selects any free primary bin $B'$ that i) has an empty space of size at least $\frac{2}{7-1/\eta}$ and ii) is not sharing an SR with any of the $f$ standby bins in the active group. If no such bin $B'$ exists, the algorithm opens a new primary bin $B'$. The bin $B'$ is then declared as the new primary bin committed to the active group, where the new primary SR is placed. 

\begin{figure*}[!t]
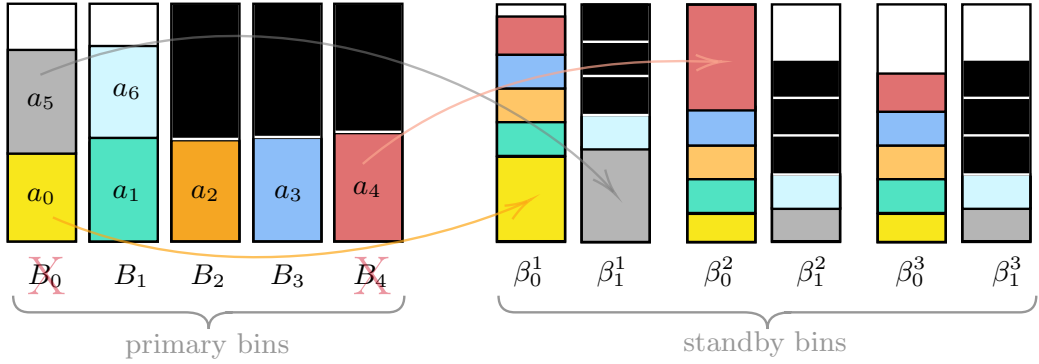

	\begin{center}
		\scalebox{1.2}{\unavailableGrp}
		\caption{An unavailable and incomplete group of bins for $(i,j)$-items, where $i=2, j=5$, and  $f=3$.}
		\label{fig:apx}
	\end{center}
\end{figure*}

\begin{example}
Figure~\ref{fig:smallSR} illustrates how \Alg packs small items, where $\eta=2$ and $f=3$. Replicas of items that form the same SR have the same color. SRs of primary replicas have sizes in the range $(\frac{1}{7-1/\eta},\frac{2}{7-1/\eta})$, that is, $(2/13,4/13)$. Standby SRs have sizes in the range 
$(\frac{1}{7\eta -1},\frac{2}{7\eta -1})$, that is, $(1/13,2/13)$. The algorithm maintains a reserved space of  $\frac{2(\eta-1)}{7\eta -1} = 2/13$ inside each standby bin. The initial small items are placed into the yellow SRs. Note that the standby SRs are placed in $f=3$ open mirroring bins ($\beta_0^1, \beta_0^2$, and $\beta_0^3$) and the primary SR is placed in bin $B_1$, which is the initially committed bin to the active group. At some point, replicas of an item $x$ do not fit in the standby SRs of the active group (the total size of primary and standby yellow replicas exceeds $\frac{4}{13}$ and $\frac{2}{13}$, respectively, if replicas of $x$ are included in the yellow SRs.). As such, a new SR (the light green SR) is opened for $x$. Given that the empty and non-reserved space in the mirroring bins of the current group is enough to fit another standby SR of size at most 2/13, the standby replicas of the new SR are placed in the open mirroring bins. Meanwhile, $B_1$ is freed, and a new bin $B_2$ is selected as the new committed primary bin to the active group, where the new primary SR is placed. Note that $B_1$ is freed because it already shares yellow SR with the standby bins of the active group. Similarly, the standby replicas of the subsequent SRs (of colors orange, blue, and pink) are placed in the mirroring bins of the active group while their primary replicas are placed in separate primary bins (each becoming the new committed bins upon freeing the previous one). At some point, replicas of some item $y$ cannot fit in the current (pink) SRs. So, a new SR (of color red) needs to be opened. The current mirroring bins do not have an available space of $\frac{2}{7\eta -1} = 2/13$; as such, the active group gets complete, and a new group with $f=3$ standby bins $(\beta_1^1, \beta_1^2$, and $\beta_1^3$) is opened, where the new standby SRs are placed. At this point, the primary replica of the new SR (the red SR) can be placed in any of the primary bins {$B_1$ to $B_6$}, that is, any of $B_1$ to $B_6$ can be selected as the committed bin to the new active group. This is because none of these primary bins are related to this new set of open standby bins. In the figure, $B_1$ is initially selected as the primary bin committed to the new group.
\end{example}

\begin{figure*}[!t]
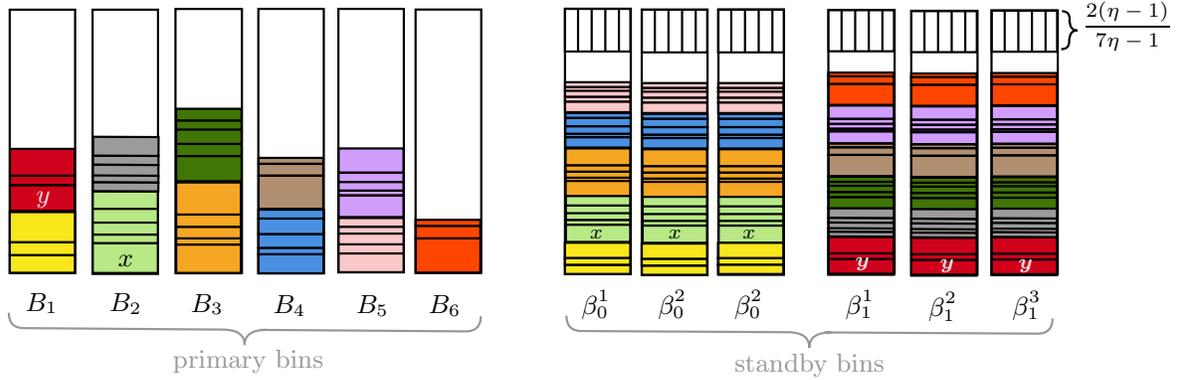

	\begin{center}
	\scalebox{1.1}{\smallSRR}
	\caption{An illustration of the \Alg packing for placing small items, where $f=3$.}
	\label{fig:smallSR}
	\end{center}
\end{figure*}
\subsection{Adjustment strategy}\label{sect:adj}
We describe an adjustment strategy that ensures a primary replica is available at any given time during the execution of \Alg. 

 Two bins in the packing of \Alg are said to be \emph{related through item $x$} if they both contain {(primary or standby)} replicas of item $x$. Clearly, any two related bins should belong to the same group of bins.  
 
 \begin{lemma}\label{lemma:help}
	In a packing maintained by \Alg, if a primary bin $B_p$ is related to a standby bin $\beta_1$ through an item $x$ and to a standby bin $\beta_2$ through an item $y\neq x$, then $\beta_1$ and $\beta_2$ are not related through any item. 
 \end{lemma}

\begin{proof}
	Let $\mathcal{P}$ be the packing maintained by \Alg.
	First, we show that (i) any pair of standby bins that are related mirror each other, and (ii) any {primary} bin in $\mathcal{P}$ shares replicas of at most one item with any standby bin.
	For (i), note that the algorithm places the standby replicas of each regular item in bins that mirror each other. The same holds for the small replicas because the standby SRs are placed into mirroring bins (see Figures~\ref{fig:algExp}~and~\ref{fig:smallSR}).  
	For small items, (ii) follows directly from the definition of the \Alg. This is because the algorithm places each primary SR into a primary bin that does not share an SR with any (standby) bin of the active group. 
	We use proof by contradiction to prove (ii) for the regular items. 
	Assume a primary bin $B$ includes primary replicas of items $x$ and $y$ of class $(i,j)$ while a standby bin $\beta'$ also includes replicas of $x$ and $y$. Given that $x$ and $y$ are regular items, they belong to the same set of the $ij$ items that are placed in the same group (with $j+fi$ bins) in $\mathcal{P}$. Let $t_x$ and $t_y$ respectively denote the indices of $x$ and $y$ in the group ($t_x \neq t_y$). Since $x$ and $y$ have their primary replicas in $B$, we should have $(t_x \mod j) = (t_y \mod j)$. Similarly, since their standby replicas are placed in $\beta'$, we should have $\lfloor t_x /j \rfloor = \lfloor t_y /j \rfloor$. This contradicts  $t_x \neq t_y$.

	Provided with (i) and (ii), we are ready to prove the lemma. Suppose the lemma does not hold, that is, a primary bin $B_p$ in $\mathcal{P}$ is related to a standby bin $\beta_1$ through an item $x$ and to standby bin $\beta_2$ through an item $y$ ($y\neq x$), while $\beta_1$ and $\beta_2$ are also related. Since $\beta_1$ and $\beta_2$ are related, by (i), they should mirror each other, that is, $\beta_1$ includes replicas of both $x$ and $y$, and so does $\beta_2$. Thus, $B_p$ shares replicas of both $x$ and $y$ with $\beta_1$ (and $\beta_2$), contradicting (ii). \qed
	\end{proof}

\noindent We use Lemma~\ref{lemma:help} to develop the adjustment strategy of \Alg:

\begin{theorem}
	\ ~ There is an adjustment strategy that ensures the packing of ~ ~ \Alg stays valid, that is, a primary replica of each item is always present in a non-failed bin.
	\end{theorem}

\begin{proof}
  
We describe an adjustment strategy that maintains an injective mapping {$h$} that maps each primary replica $x$ placed originally in a failed primary bin into a {non-failed} standby bin {$h(x)$} which hosts a standby replica of the same item. 
For each primary replica $x$ in a failed primary bin, the standby replica in {$h(x)$} replaces $x$ as the new primary replica. Since the mapping is injective, at most one replica in each standby bin will be promoted to a primary replica. Given that each standby bin with a replica of size $s$ has an empty space of at least $(\eta-1)s$, no bin is overloaded. In what follows, we describe an adjustment strategy that maintains the desired injective mapping as bins fail and recover. In this process, the standby bins that are in the range of {$h$} are referred to as ``marked" bins.

We describe how to maintain the mapping $h$ at time $t$.  
Suppose that such mapping is maintained in the previous $t-1$ steps. Let $f_p$ denote the number of primary bins that are failed (and not recovered) before time $t$ (we have $0 \leq f_p \leq f$). Suppose that out of these $f_p$ failed primary bins, $r_p$ bins are recovered at time $t$ ($0\leq r_p \leq f_p$). A bin $B$ is said to be \emph{critical} iff it fails at time $t$ while containing a primary replica.
Let $k$ denote the number of {critical} bins.
All primary bins that fail at time $t$ are critical ({non-failed} primary bins contain primary replicas). Marked standby bins 
that fail at time $t$ are also critical (they are in the range of {$h$} and hence a replica in them has replaced the primary replica of a failed primary bin). 
For the packing to stay valid, the primary replicas in the critical bins should be mapped to some {non-failed} bins. 

The adjustment algorithm first ensures that the primary replicas in the $r_p$ recovered primary bins are removed from the domain of {$h$}, and retain their primary status; this means the standby replicas in the range of $h$ that were previously upgraded to primary replicas become standby again and their bins become unmarked. At this point, the number of failed, unmarked standby bins is at most $f - (f_p-r_p+k)$; this is because, out of at most $f$ failed bins, $f_p-r_p$ of them are primary bins that are failed before $t$, and $k$ of them are critical and hence are either primary bins or marked standby bins.  

Consider an arbitrary ordering $(B_1, B_2, \ldots, B_{k})$ of the critical bins. We process the critical bins, one by one, in this order. When processing a bin $B_q$, we process primary replicas in {$B_q$} in an arbitrary order ({$q\leq k$}). Let $a$ be a replica in $B_q$ that is being processed, and let $A$ be the set of $f$ standby bins that include replicas of $a$. 
We need to map $a$ to an unmarked bin $\beta \in A$ and then mark $\beta$. We show that it is always possible to find such an unmarked bin $\beta$. 
Consider a previously processed bin $B'$ {(whose primary replicas are mapped)}, that is, either $B'$ failed previously at $t'<t$ or $B'$ is {$B_{q'}$} for {$q'<q$}. We claim that {during the process of} $B'$, at most one bin from $A$ has become marked. 
At the time $B'$ is processed, it has been critical and hence either a primary bin or a marked standby bin. 
If $B'$ was a {marked} standby bin, then it contained at most one primary replica (since the mapping is injective), and the claim holds. To prove the claim when $B'$ is a primary bin, consider otherwise, that is, assume two standby bins $\beta_x,\beta_{y} \in A$ {have been} marked {during the process of $B'$} . This means 
$B'$ is related to $\beta_x$ and $\beta_{y}$ through two different items. {On the other hand,} $\beta_x$ and $\beta_{y}$ are also related to each other (because they are both in $A$ and hence contain a replica of $a$). This is not possible, however, given the result in Lemma~\ref{lemma:help}. 

There are $f_p-r_p+q-1$ failed bins that are processed before $B_q$. By the above argument, processing any of these bins results in marking at most one standby bin from $A$. Therefore, at the time of processing $a$, at most $f_p-r_p+q-1$ standby bins from $A$ are previously marked. Note that $a$ is replicated on $f$ standby bins. As a result, there are at least $f-f_p+r_p-q+1$ unmarked standby bins that host standby replicas of $a$. Among these bins, at most $f-(f_p+r_p-k)$ bins are failed. So, there are at least $f-f_p+r_p-q+1 - (f-f_b+r_p-k)  = k-q+1$ non-failed and unmarked standby bins in~$A$. 
Given that {$q \leq k$}, there is at least one unmarked bin that $a$ can be mapped to.  \qed

\end{proof} 

\section{Competitiveness of \Alg}\label{sect:analysis}

In this section, we use a weighting argument to provide an upper bound for the competitive ratio of \Alg. 
We assign a \emph{weight} to each replica in the final packing of the algorithm. The weights are defined in a way that the total weight of replicas placed in each bin of the algorithm, except possibly a constant number of them, is at least 1. Therefore, if $w(\sigma)$ denotes the total weight of all replicas in the input sequence, the number of bins in the packing of \Alg is no more than $w(\sigma)+c$, for some constant $c$ independent of the input length (but possibly a function of parameters $f$ and $\eta$). At the same time, we show that any bin in an optimal packing has a weight at most $1.75$, which means the number of bins in an optimal packing is at least $w(\sigma)/1.75$. As such, the competitive ratio of \Alg will be at most $1.75$. 

\noindent \paragraph{{Weighting:}} For regular items, define the weight of a primary replica of class $i \ (\leq 6)$ as $1/i$, and the weight of standby replicas of class $j \ ( \leq \lfloor 6\eta \rfloor)$ as $1/j$. 
For small items, a primary or standby replica of size $x$ has weight $3x/2$ (see Table~\ref{tbl:weighting}).

\begin{lemma}\label{lemma:algWeight}
	The total weight of replicas in any bin of \Alg, except for at most a constant number of bins, is at least 1.
\end{lemma}
\begin{proof}
First, we investigate regular bins. {Consider the bins in the complete groups (see Figure~\ref{fig:algExp}). In any such group,} a primary bin of class $i \leq 6$ includes $i$ replicas, each of weight $1/i$. Similarly, a standby bin of class $j \leq \lfloor 6\eta \rfloor$ includes $j$ replicas, each of weight $1/j$. Therefore, all {regular} bins, except for those in the incomplete groups, have weight~1. 
{We show that the total number of bins inside incomplete groups is a constant independent of the input length.} By Lemma~\ref{lem:new}, there are at most $f+1$ incomplete groups for $(i,j)$-items. There are  
$j+fi \leq 6(\eta +f)$ bins inside each group. So, there are at most $6(f+1)(\eta +f)$ partially-filled bins for $(i,j)$-items. Given that $i\leq 6$ and $j\leq 6\eta$, there are at most $36\eta$ possible pairs of $(i,j)$. In total, the number of partially filled bins for regular items is at most $(36\eta) 6(f+1)(\eta +f) = O(1)$. 
In summary, the total weight of items in any regular bin, except for at most $O(1)$ of them, is at least~1.

Next, we look into small items. Let $x$ be the primary SR that causes opening the last primary small bin. Also, let $m$ be the number of standby SRs packed in the standby bins of the active group at the time $x$ was {placed}. 
There are $m$ primary bins that are related to the $f$ standby bins, 
and thus cannot host $x$. By Lemma~\ref{lem:new}, there are up to $f$ incomplete groups other than the active group. The primary bins committed to these groups also cannot host $x$. 
The remaining primary bins could not host $x$ only because they did not have enough space. So, all primary small bins, except at most $m+f$ of them, are filled to a level of {at least} $1-\frac{2}{7-1/\eta} = \frac{5\eta-1}{7\eta-1} \geq 2/3$. 
 Given that any small item of size {$s$} has weight {$1.5s$}, the total weight of replicas in any of these bins is at least $ 1.5 \ (2/3) = 1$. Next, we show that $m$ is a constant with respect to the input length.
 Each standby bin has a non-reserved space of size $1-\frac{2(\eta-1)}{7\eta-1} = \frac{5\eta+1}{7\eta-1}$ which is used to pack SRs of size at least  $\frac{1}{7\eta-1}$. As such, we have $m \leq 5\eta+1$. {So, all primary small bins, except for at most $5\eta+1+f \in O(1)$ of them, have weight at least~1.} 
Standby small bins have a reserved space of $\frac{2(\eta-1)}{7\eta -1}$.  Except for the bins inside the incomplete groups, other bins have an additional empty space of at most $\frac{2}{7\eta -1}$, giving them a total empty space of {at most} $\frac{2\eta}{7\eta -1}$. 
By Lemma~\ref{lem:new}, there are up to $f+1$ incomplete groups, each containing $f$ mirroring bins. 
Therefore, the filled space in each {standby} bin, {except for at most $f(f+1) = O(1)$ of them,} is at least $\frac{5\eta-1}{7\eta -1}\geq 2/3$. {Given that any standby small replica of size $s$ has weight {$1.5s$},} the weight of any of these standby bins is then at least $1.5 \ (3/2)  = 1$.
%
\qed
\end{proof}

\vspace*{1mm} 
\begin{lemma}\label{lemma:optWeight}
	The total weight of items in any bin of an optimal packing is at most 1.75.
\end{lemma}

\begin{proof}
	Define the \emph{density} of each item as the ratio between the weight and the size of the item.
	A primary replica of class $i\leq 6$ has a size in the range \mbox{$(\frac{1}{i+1}, \frac{1}{i}]$} and weight $1/i$, which gives a density of at most $\frac{i+1}{i}$. 
	Similarly, standby replicas of class $j\leq \lfloor 6\eta \rfloor$ have size in the range \mbox{$(\frac{1}{j+\eta}, \frac{1}{j+\eta-1}]$} and weight $1/j$, giving them a density of at most $(j+\eta)/j$.
	Small replicas (both primary and standby) have a density of 3/2. 
	We consider three possible cases and show that the total weight of items in any bin $B^*$ of an optimal packing is at most 1.75 in each case. 
	To follow the proof, it helps to consult  Table~\ref{tbl:weighting}.
    \vspace{3mm}
    
	\noindent \textbf{case 1: no standby replica in $B^*$:  } 
    Suppose $B^*$ does not include any standby replica. 
    In this case, $B^*$ includes 0 or 1 primary replica of class~1 (it cannot include more than 1 such replica since all replicas of class~1 have sizes larger than 1/2). If it includes no replica of class~1, the density of each of the items (of other classes) is at most 1.5, giving a total weight of at most 1.5 for items in $B^*$. If $B^*$ includes one item of class~1 (with weight 1), the total size of other items will be less than 1/2, and since their density is at most 3/2, their total weight will be no more than $3/2 \cdot 1/2 = 3/4$, giving a total weight of at most $1+3/4=1.75$ for items in $B^*$. 
    \vspace{3mm}
    
    \noindent \textbf{case 2: some standby replicas in $B^*$ are regular:\  } 
	Suppose $B^*$ includes at least one regular standby replica. Let $x$ be the largest standby replica in $B^*$ and $j$ denote the class of $x$; we have $1\leq j \leq \lfloor 6\eta \rfloor$. 
	There should be enough empty space in $B^*$ so that if {all $f$ bins} containing other replicas of $x$ are failed, $x$ can be declared as a primary replica. Increasing the size of $x$ by a factor $\eta$ should not cause an overflow, that is, there should be empty space of at least $(\eta-1)x > (\eta-1)/(j+\eta)$ in $B^*$ (recall that replicas of class~$j$ are of sizes at least $1/(j+\eta)$). So, the total size of items in $B^*$ is {less than} $1-\frac{\eta-1}{j+\eta}=\frac{j+1}{j+\eta}$. 
	There are two cases to consider: either $j=1$ or $j \geq 2$: \vspace{1mm}
	
	\noindent \textbf{i)} Suppose $j=1$, that is, there is a standby replica $x$ of size more than $1/(1+\eta)$ in $B^*$. 
	The total size of items in the bin {is less than} $\frac{j+1}{j+\eta} = \frac{2}{\eta+1}$, and items other than $x$ in $B^*$ have a total size {less than} $\frac{1}{\eta+1}$. {As a result, there is no primary replica of class~1 (of size at least $1/2>\frac{1}{\eta+1}$) or standby replica of class~1 (of size more than $\frac{1}{\eta+1}$) in $B^*$.}
	So, primary replicas in $B^*$ have class~2 or more and hence density at most 1.5. Similarly, standby replicas other than $x$ have class~2 or more and hence density no more than {$\frac{\eta+2}{2}$}. So, all replicas other than $x$ in $B^*$ have a density at most {$\max\{1.5, \frac{\eta+2}{2}\}=$} $\frac{\eta+2}{2}$. 
	Since the total size of these replicas is at most $\frac{1}{\eta+1}$, their total weight is at most {$\frac{1}{\eta+1} \cdot \frac{\eta+2}{2} = \frac{\eta+2}{2\eta+2}$}. 
	{Adding the weight 1 of $x$}, the total weight of replicas in $B^*$ is at most {$\frac{3\eta+4}{2\eta+2}$}, which is at most 1.75, given that $\eta >1$. 
	\vspace{1mm}
	
	\noindent \textbf{ii)} Suppose $j\geq 2$. Recall that the total size of items in $B^*$ is less than $\frac{j+1}{j+\eta}$. First, assume there is also a primary replica $y$ of class~1 in $B^*$. This is possible only if $\frac{1}{j+\eta}+\frac{1}{2}<\frac{j+1}{j+\eta}$, that is, $\eta<j$. The size of replicas other than $y$ in $B^*$ is less than $\frac{j+1}{j+\eta} - 1/2=\frac{j+2-\eta}{2j+2\eta}$, and their density is at most $\max\{1.5, (j+\eta)/j\}$ (primary replicas of class $\geq 2$ have density at most 3/2, and standby replicas have density at most $(j+\eta)/j$). The total weight of replicas other than $y$ is hence {less than} $\max\{\frac{3j+6-3\eta}{4j+4\eta}, \frac{j+2-\eta}{2j}\} \leq \max\{3/4, \frac{j+2-\eta}{2j}\}$, which is at most 0.75, given that $\eta >1$ and $j\geq 2$. Adding the weight 1 of $y$, {the total weight of replicas in $B^*$ will not be more than $1.75$}.
	Next, assume there is no primary replica of class~1 in $B^*$. In this case, the total size of replicas in $B^*$ is at most $\frac{j+1}{j+\eta}$, and their density is at most {$\max\{1.5, (j+\eta)/j\}$}, giving them a total weight of at most $\max\{\frac{3j+3}{2j+2\eta}, (j+1)/j\}$, which is at most $\max\{1.5, (j+\eta)/j \}$ = 1.5, given that $\eta>1$ and $j\geq 2$.
	\vspace*{3mm} 
	
	\noindent \textbf{case 3: all standby replicas in $B^*$ are small: \ }
	Assume there is no regular standby replica in $B^*$, but there is at least one small standby replica in $B^*$. We consider two cases: either there is a primary replica of class~1 in $B^*$ or not:
	\vspace*{1mm} 

	{\noindent \textbf{i)} If there is a primary replica $y$ of class~1 in $B^*$, the remaining space of $B^*$ is less than 1/2 (as $y$ is of size more than 1/2). No other primary replica $y'$ of class~1 can be in $B^*$ because each of $y$ and $y'$ would have a size more than 1/2. Therefore, the remaining space in $B^*$ (of size less than 1/2) can be filled with primary replicas of class $i\geq 2$ (of the density of at most 3/2) and with other standby small replicas (of density 3/2). So, the total weight of items other than $y$ in $B^*$ is at most $1.5 (1/2) = 3/4$. Given that the weight of $y$ is 1, the total weight of items in $B^*$ will be at most $1.75$.}
	\vspace*{1mm} 
	
	\noindent \textbf{ii)}  If there is no primary replica of class~1 in $B^*$, then $B^*$ is filled with primary replicas of class $i \geq 2$ (of density at most 3/2) and standby replicas of class $j=\lfloor 6\eta \rfloor + 1$ (of the density of 3/2). As a result, the total weight of $B^*$ will be no more than $3/2$.
	\qed
\end{proof}
\vspace*{2mm} 
\begin{theorem}
     \Alg has a competitive ratio of at~most~1.75.
\end{theorem}
\begin{proof}
Let $\sigma$ be any input sequence, and $w(\sigma)$ be the total weight of items in~$\sigma$. Let $HS(\sigma)$ be the number of bins that \Alg opens for $\sigma$. By Lemma~\ref{lemma:algWeight}, we have $HS(\sigma) \leq w(\sigma)+c$ for some constant $c$ independent of $|\sigma|$. On the other hand, by Lemma~\ref{lemma:optWeight}, we have $\opt(\sigma) \geq w(\sigma)/1.75$. We can write $\frac{HS(\sigma)}{Opt(\sigma)} \leq \frac{w(\sigma)+c}{w(\sigma)/1.75}$, which converges to 1.75, given that $c$ is a constant. \qed
\end{proof}

\section{Concluding remarks}

We proved that the competitive ratio of \Alg is at most 1.75, which is an improvement over the competitive ratio 2 of the best existing algorithm. We note that this upper bound 
holds for all values of $f$ and $\eta$. 
When $\eta$ is close to 1, the existing lower bounds for the classic online bin packing extend to the fault-tolerant setting. In particular, no fault-tolerant bin packing algorithm can achieve a competitive ratio better than 1.54~\cite{BalBek12,BaloghBDEL19}. 
As a topic for future work, one may consider tightening the gap between the lower bound of 1.54 and and the upper bound 1.75. 

\section*{Acknowledgement}
We acknowledge the support of the Natural Sciences and Engineering Research Council of Canada (NSERC).

\bibliographystyle{splncs04}
\bibliography{refs}

\newpage

\end{document}